\documentclass[a4paper,12pt]{amsart}

\usepackage{kantlipsum} 
\calclayout

\usepackage{amsmath}    
\usepackage{amssymb}    
\usepackage{bm}         
\usepackage{enumitem}
\usepackage{graphicx}   
\usepackage{hyperref}   
\usepackage{natbib}    
\usepackage{epsf}
\usepackage{lscape}
\bibpunct{(}{)}{;}{a}{,}{,}

\newtheorem{prop}{Proposition}

\usepackage{multirow}

\newcommand{\ATE}{\text{ATE}}
\newcommand{\TE}{\text{TE}}
\newcommand{\E}{\text{E}}

\begin{document}

\title{%
  Specification analysis for technology use and teenager well-being:
  statistical validity and a Bayesian proposal
}

\author{Christoph Semken$^{1,2}$ and David Rossell$^{1,2}$
\\$^1$: Universitat Pompeu Fabra, Barcelona, Spain
\\$^2$: Barcelona School of Economics, Barcelona, Spain}


\maketitle

\begin{abstract}
  A key issue in science is assessing robustness to data analysis choices, while avoiding selective reporting and providing valid inference. Specification Curve Analysis is a tool intended to prevent selective reporting. Alas, when used for inference it can create severe biases and false positives, due to wrongly adjusting for covariates, and mask important treatment effect heterogeneity. As our motivating application, it led an influential study to conclude there is no relevant association between technology use and teenager mental well-being. We discuss these issues and propose a strategy for valid inference. Bayesian Specification Curve Analysis (BSCA) uses Bayesian Model Averaging to incorporate covariates and heterogeneous effects across treatments, outcomes and sub-populations.
  BSCA  gives significantly different insights into teenager well-being, revealing that the association with technology differs by device, gender and who assesses well-being (teenagers or their parents). 
  
\end{abstract}

\section{Introduction}

Choosing an appropriate statistical model is an old and important question 
that received renewed attention in light of the reproducibility crisis that plagues many research fields \citep{begley_drug_2012,open_science_collaboration_estimating_2015,baker_is_2016,camerer_evaluating_2018}. ``P-hacking'' and selective reporting – the practice of picking the model or analysis that maximizes a hypothesized effect is often seen as a main driver \citep{benjamin_redefine_2018}. A potential solution to selective reporting is to obtain results under many models, covering many analytical choices 
\citep{athey_measure_2015,young_model_2017}, so that one can assess the sensitivity of results to the chosen statistical model.
The question is then how to either select what models to report or aggregate them into a main finding.

One methodology that aims to address this question is Specification Curve Analysis (SCA, \cite{simonsohn_specification_2020}). 
SCA considers the situation where one wishes to study the association between several outcomes and several treatments of interest (which we refer to as {\it treatment effects}), considering potential covariates (which we refer to as {\it controls}) to adjust the analysis for, and also considering several possible sub-populations of individuals that one could focus the analysis on.
Here, selective reporting could occur if a researcher were to only report results for a particular outcome, treatment or sub-population that supports a pre-conceived finding, failing to indicate that other treatments, outcomes or sub-populations do not support the finding.
SCA attempts to prevent selective reporting by presenting the estimated treatment effects under each possible model, i.e. obtained by regressing each possible outcome on each possible treatment, using all possible control covariate combinations and all sub-populations of individuals.
SCA plots these results in a single display that can serve as a descriptive sensitivity analysis.
For example, Figure \ref{fig:sca} shows the estimated associations between two technologies (TV and electronic device use) and 5 outcomes (loneliness and four related to suicide), both using and not using control covariates. The plot reveals that not all models lead to statistically significant associations, and that their magnitude varies quite a bit.
A nice feature of SCA is urging caution in such situations: if one were to report only one of these estimates, then one would have to carefully justify that choice.
However, as a problematic issue discussed here, SCA also performs a formal hypothesis tests on the median of all these effects (or their sign). As argued below, such a test is a poor strategy to aggregate results that can lead to statistically invalid conclusions.

SCA has been used in many fields, including psychology \citep{rohrer_probing_2017,bryan_replicator_2019-1,hassler_large-scale_2020}, political science \citep{dunning_information_2019}, economics \citep{cookson_when_2018,lejarraga_no_2019} and neuroscience \citep{cosme_multivariate_2020}.
It has also been covered by a number of methodological guides and reviews (e.g., \citealt{christensen_transparency_2018-1,forstmeier_detecting_2017,george_big_2016,milfont_replication_2018,orben_teenagers_2020,simmons_false-positive_2018,weston_recommendations_2019,wuttke_why_2019}).
It is therefore important to find a way to ameliorate certain pitfalls in SCA, while maintaining its value as a tool to explore sensitivity to analysis choices.

A main motivating application behind our developments is an influential study by \cite{orben_association_2019} on the association between technology use and teenager well-being.
SCA led these authors to conclude that ``the association of [adolescent mental] well-being with regularly eating potatoes was nearly as negative as the association with technology use''. That is, said association does not have a practical significance.
These findings, which were portrayed in outlets such as \citet{the_new_york_times_is_2019} and \citet{forbes_magazine_screen_2019}, have important consequences for the public's perception and decision makers.
Unfortunately, such conclusions are due to using inadequate statistical methodology: SCA fails to account for control covariates and combines heterogeneous effects across treatments and outcomes into an overall summary that can be severely misleading.
Here we develop an extension of SCA that addresses these issues and, when applied to the data of \cite{orben_association_2019}, gives very different conclusions.
There exist associations between certain technologies and well-being of high practical relevance, and there are stark differences between these associations as assessed by parents versus teenagers.
In view of increased demands that social media platforms should be accountable for the well-being of its users, particularly potentially vulnerable members such as teenagers, it is critical that the debate on such technologies is informed by sound statistical methodology.
We remark, however, that our analyses do not imply causal relationships between technology and well-being, rather they inform about their conditional association after one accounts for several controls. For example, it may be that sadness leads one to use more technology, rather than the other way around.

We elaborate on the two mentioned pitfalls of SCA.
First, it fails to properly account for controls. Rather than using the controls that are likely to have an effect on the outcome, based on the observed data,
the SCA median weights all possible control configurations equally. For example, suppose that there is a single control and there is strong evidence that it is associated with the outcome, so that it needs to be included in the regression to avoid an omitted variable bias. SCA would obtain the median of the estimated effect when including the control and when excluding it, and would hence run into the omitted variable bias, which as we illustrate below can lead to a 100\% false positive rate even in simple situations.
Second, by reporting median effects over different treatment-outcome combinations and across multiple sub-populations, SCA can mask critical heterogeneity of the treatment effects. If, for example, a treatment has opposite effects on two different outcomes their median can be essentially zero, e.g. as is the case for certain parent versus teenager assessments of well-being.
In layman's terms, an average treatment effect (ATE) may appear practically irrelevant due to averaging over apples and oranges.
By conducting a hypothesis test on a single aggregate estimate, SCA cannot detect such heterogeneous effects.

In this paper, we propose an alternative aggregation method to address these issues – the Bayesian Specification Curve Analysis (BSCA).
BSCA uses Bayesian Model Selection and Averaging (BMS and BMA respectively, see for example \cite{hoeting_bma_1999,clyde_introduction_2020}) to produce separate estimates for each treatment-outcome combination  and hence acknowledge their heterogeneity.
By using a convenient parameterization it also allows to consider sub-populations where their effects may deviate from the average treatment effect, i.e. interactions between the treatments and the sub-population indicator.
To avoid controversies in setting prior parameters, we use an approximation to BMA given by the Extended Bayesian Information Criterion (EBIC; \citealt{chen_extended_2008}). The EBIC sets stringent thresholds for including parameters in the model, hence helping prevent false positives and selective reporting.
We illustrate how our approach, which is a relatively direct extension of well-established methodology, leads to vastly improved statistical properties, including lower bias and type I error rates.

Our aim is not to provide a detailed theoretical critique of SCA and similar aggregation methods, which can be found elsewhere \citep{slez_difference_2019,giudice_travelers_nodate}. Instead, we seek to provide a practical alternative that relies on statistically principled methods and can address the selective reporting issues that motivated SCA (considering multiple models defined by choosing treatments, controls, outcomes and sub-populations). 
For completeness, the Supplementary Materials include an introduction to the Bayesian framework and a BSCA tutorial for practitioners.
Our secondary goal is to provide cues that can help inform the debate on teenager technology use in a statistically-principled manner. While space does not allow us to fully discuss the extensive related behavioral literature, we do relate some of our findings to said literature in Section \ref{sec:discussion}, particularly where it helps emphasize that the methodological issues discussed here can have important practical consequences.

The paper is structured as follows. Section \ref{sec:methods} formalizes the data analysis problem, reviews SCA and describes our BSCA proposal, along a parameterization that allows to naturally incorporate subgroup analysis.
We first describe a single-outcome BSCA to study the effects of multiple treatments on one given outcome, both on average and within sub-populations.
We next describe a multiple-outcome BSCA which allows one to visualize the effect of multiple treatments on multiple outcomes, as well as global average treatment effects across outcomes.
Finally, we discuss how one may present association measures other than regression coefficients, if so desired.
Section \ref{sec:simulation} illustrates the poor statistical properties of SCA in simple simulations, and the improvements brought by BSCA. It also illustrates how our EBIC-based formulation helps prevent false positives, even when conducting multiple hypothesis tests.
Section \ref{sec:application} introduces the teenager data and discusses how the BSCA findings contribute to the debate on technology use and teenager well-being. 
Section \ref{sec:discussion} concludes.
The appendix contains the proof of a simple proposition, and we provide further results and R code as supplementary material.

\begin{figure}
  \includegraphics[width=\linewidth, height=0.6\linewidth]{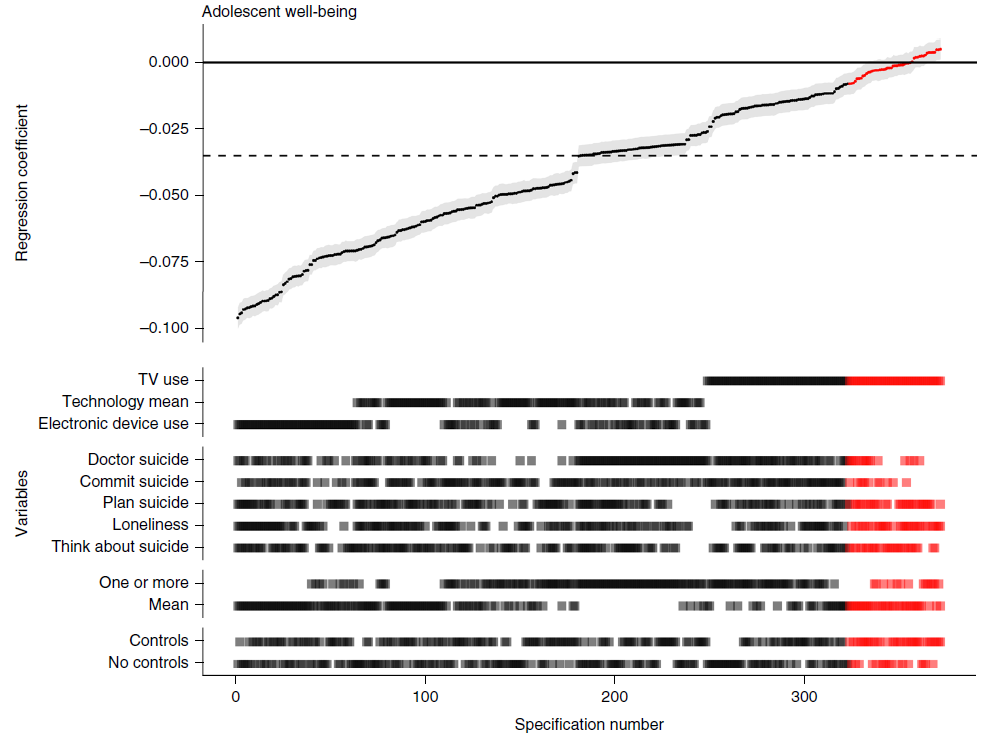}
  \caption{Specification curve analysis, reproduced from \citet{orben_association_2019}. Top: estimated treatment effect and 95\% confidence interval of technology use on teenager well-being, obtained from a linear regression model. Dotted line: median estimate. Red markers: statistically non-significant results (P-value $>$0.05). Bottom: independent and dependent variable(s) used. ``One or more'' indicates whether the dependent variable was defined to be 1 when $\geq 1$ of the 5 outcome variables was 1, or the mean of said 5 outcomes. ``Controls'' indicates whether race was included as a control variable or not. Data source: YRBS. See Section S4 for details and possible issues.}
  \label{fig:sca}
\end{figure}

\section{Methods} \label{sec:methods}

Consider a setting where there are $L$ outcomes $y^{(l)} \in \mathbb{R}^n$ for $l=1,\ldots,L$, with $n$ observations for each, that one may use to study a phenomenon (e.g. several outcomes measuring teenager well-being).
The goal is to learn the association between the outcomes and a set of $J$ treatments, recorded in an $n \times J$ matrix $X$ (e.g. social media use, internet use, TV use).
One wishes to adjust for $Q$ controls recorded in an $n \times Q$ matrix $Z$ (e.g. household or socio-economic characteristics),
and also consider that said association may differ across $K$ sub-populations (e.g. defined by gender or age).
The subgroup information is coded in an $n \times K$ matrix $G$, using a parameterization specified in Section \ref{sec:methods_model}.
The reason for considering sub-populations is that one suspects that treatment effects could differ across sub-populations, i.e. there are heterogeneous treatment effects, which can be statistically formalized as there being an interaction between the treatment and the sub-population. If said interaction is not present, then one would rather remove it from the model and focus on the average effect for the whole population.
We denote by $x_i \in \mathbb{R}^J$ the $i^{th}$ row in $X$, $z_i \in \mathbb{R}^Q$ that in $Z$, and $g_i \in \mathbb{R}^K$ that in $G$.

We outline a generic generalized linear model to study the treatment effects,
for simplicity focusing first on the setting with a single outcome $y=(y_1,\ldots,y_n)$.
The regression equation is
\begin{equation} \label{eq:problem_sca}
 F \left( E(y_i \mid x_{i}, g_{i}, z_{i}) \right) = \alpha + \gamma^T z_{i} + \eta^T g_{i} + \sum_{j=1}^J ( \beta_{j} + \delta_{j}^T g_{i}) x_{ij}
\end{equation}
where $\alpha \in \mathbb{R}$ is the intercept, $\gamma \in \mathbb{R}^Q$ are regression coefficients for the controls,
$\eta \in \mathbb{R}^K$ the main effects for the subgroups,
$\beta=(\beta_1,\ldots,\beta_J)$ are treatment effects and $\delta_j \in \mathbb{R}^K$ their modifiers for the sub-populations defined by $g_i$.
Importantly, our parameterization in Section \ref{sec:methods_model} ensures that $\beta_j$ can be interpreted as the average treatment effect across subpopulations.
The total number of regression parameters is $p=1+Q+K+J(1+K)$.
$F$ is a link function specifying the functional form linking the outcome and covariates, e.g. the logistic function in logistic regression.
For simplicity we focus the exposition on Gaussian linear regression where $F$ is the identity, but our methodology is also implemented for other generalized linear models, and our teenager application includes logistic regression examples.

The issue is that it is not clear a priori whether all terms in (\ref{eq:problem_sca}) are really necessary, i.e. what controls should one adjust for, is it justified to consider sub-populations, or should one only consider a subset of treatments.
Section \ref{ssec:sca} reviews how SCA addresses this issue and discusses the main features in our BSCA proposal.
Section \ref{sec:methods_model} presents our parameterization for the treatment effects in (\ref{eq:problem_sca}).
Section \ref{sec:methods_bma} discusses how BSCA is built on Bayesian model selection and averaging under an approximation given by the EBIC.
Section \ref{sec:methods_single} describes how to perform BSCA for a single outcome, and Section \ref{sec:methods_multiple} how to summarize the results across multiple outcomes in a single display.
Finally, Section \ref{ssec:alternative_measures} discusses how one may perform BSCA inference for measures of association other than regression coefficients.

\subsection{Specification curve analysis}
\label{ssec:sca}

SCA requires the researcher to identify a ``set of theoretically justified, statistically valid and non-redundant'' analysis strategies, called specifications.
Each of these defines a sub-model of (\ref{eq:problem_sca}), which can vary in terms of what outcome, treatments and controls are considered,
and for what subgroups (if any) one wishes to obtain separate treatment effect estimates for.
Denote by $s=1, \ldots, S$ the possible sub-models, by $(x_{si}, g_{si}, z_{si})$ the corresponding subset of covariates $(x_i,g_i,z_i)$ selected by model $s$,
and by $\hat{\beta}_{sj} + \hat{\delta}_{sj}^T g_i$ the estimated treatment effects.

SCA estimates the effect of each treatment featuring in each considered model $s$ (or a random sample of models, if there are too many to enumerate fully),
and plots all these estimated $\hat{\beta}_{sj} + \hat{\delta}_{sj}^T g_i$ into a ``descriptive specification curve''.
The idea is that one may easily visualize how the estimates vary according to what specific outcome, treatment, controls or individual subgroups were used.
Figure \ref{fig:sca} shows an example from \citet{orben_association_2019}. The top panel shows the estimated coefficients (sorted increasingly) and the bottom panel which 
treatment was used (electronic device use, TV use, or their average), 
what outcomes (either loneliness or four related to suicide), how the multiple outcomes were aggregated into a scalar response, and whether a vector of control variables was included or not.
The red color is used to highlight what coefficients received a P-value above 0.05 in their individual models.

Besides this descriptive use, \cite{simonsohn_specification_2020} also proposed to obtain a global effect estimate and to test its statistical significance via a hypothesis test based on a type of permutation and paired bootstrap procedures. 
They propose three such global estimates:
(i) the median of all estimated effects, (ii) the ``share of specifications
that obtain a statistically significant effect in the predicted direction'' or (iii) the average Z-test statistic value.
Most SCA applications that we are aware of, including \citet{orben_association_2019} use the median effect. We hence focus our discussion on this measure, but similar considerations apply to the other two.

As discussed earlier, there are two main issues that render the median treatment effect $\hat{\beta}_{sj} + \hat{\delta}_{sj}^T g_i$ across all models $s$ statistically invalid, in the sense of not estimating the true (conditional) outcome-treatment association consistently even as $n \rightarrow \infty$.
First, one assigns the same weight to all control configurations. For example, with a single control one would take the median of the estimated effect when including and when excluding the control, regardless of the statistical (nor practical) significance of the control's estimated coefficient.
It is well-known that failing to adjust for controls can seriously bias the estimates, inflate their variance, and result in serious issues for the type I error of the test.
See Section \ref{sec:simulation} for a simple example where the SCA median test has a type I error of 1.
Second, computing the median (or average) effect over all treatments and outcomes can mask important heterogeneity.
For example, in the teenager data averaging over electronic device and TV use, or across teenager and parent assessments, led \citet{orben_association_2019} to conclude that the association between technology and teenager well-being is not practically relevant.

Our proposal is based on using BMS and BMA to infer treatment effects in (\ref{eq:problem_sca}), while incorporating the uncertainty regarding what controls are most appropriate and the potential existence of heterogeneous treatment effects across subpopulations.
This is similar in spirit to SCA, which also averages over control configurations and subsets of individuals. 
The critical difference is that BMA weights each control configuration based on its posterior probability given the observed data, and also uses posterior probabilities to evaluate heterogeneity across subpopulations.
In contrast, SCA aggregates models using equal weights that are specified a priori, and are hence not informed by the data.
Another important difference relative to SCA is that, by default, we do not perform a single test on the average effects over treatments and outcomes, but rather portray their heterogeneity. We remark that, if desired, our framework also allows testing such averages, as discussed in Section \ref{sec:methods_multiple}.

\subsection{Model parameterization} \label{sec:methods_model}

We discuss how to ensure that the parameters $\beta_j$ quantifying treatment effects in (\ref{eq:problem_sca}) have the same interpretation in any of the $S$ considered models, else the reported average may be non-sensical.
Suppose that Model 1 does not include interactions with subgroups ($\delta_j=0$), whereas Model 2 includes them ($\delta_j \neq 0$). Unless both models are suitably parameterized, the interpretation of $\beta_j$ would be different under each model. 
In particular, if one were to use the standard binary indicators to code for treatments and subgroups in Model 2, then $\beta_j$ would be the treatment effect for the reference subgroup, whereas for Model 1 $\beta_j$ is the average treatment effect across all individuals.

Fortunately, there is a simple parameterization such that $\beta_j$ gives the average treatment effect across the $n$ individuals, regardless of whether the model includes interactions with subgroups or not.
The solution is to code the treatments as $x_{ij} \in \{-1/2,1/2\}$, where $x_{ij}=1/2$ indicates that individual $i$ received treatment $j$ and $x_{ij}=-1/2$ otherwise. Regarding the subgroups coded into $g_i \in \mathbb{R}^K$, let $\rho_k$ be the proportion among the $n$ individuals that belong to group $k=1,\ldots,K$, we then code
\begin{align}
  g_{ik} = \begin{cases}
    \rho_k \text{, if the individual belongs to subgroup } k \\
    -(1 - \rho_k) \text{, otherwise}
  \end{cases}. \qquad
\nonumber
\end{align}
Under this parameterization, simple algebra shows that 
$\sum_{i=1}^n g_i=0$.
Hence, the effect of treatment $j$ for individual $i$ is
\begin{equation}\label{eq:te}
  \TE_{ij} = \E(y_i \mid x_{ij} = 1/2, g_i) - \E(y_i \mid x_{ij} = -1/2, g_i)=
  \beta_j + \delta_{j}^T g_{i}
\end{equation}
and the average treatment effect across all individuals is
\begin{align}\label{eq:ate}
\ATE_j= \frac{1}{n} \sum_{i=1}^n \TE_{ij}= 
\beta_j + \frac{1}{n} \delta_j^T \sum_{i=1}^n g_i=
\beta_j.
\end{align}

That is, the parameterization enforces a sum-to-zero constraint, such that $\beta_j$ is the ATE and $\delta_j^T g_i$ gives the deviations from the ATE for each sub-population.
Expression (\ref{eq:ate}) obviously remains valid in models that do not include heterogeneity of effects across subgroups ($\delta_j=0$), and can be easily extended to the case where there are multiple subgroup indicators.
For non-Gaussian generalized linear models where the link function $F$ in (\ref{eq:problem_sca}) is not the identity,
we define the treatment effect in terms of the linear predictor, so one obtains the same expression as in \eqref{eq:te}. For example, in logistic regression the treatment effect in \eqref{eq:te} is defined in terms of the log odds-ratio.

We remark that in this paper we treat both the ATEs given by $\beta_j$ and the sub-group specific parameters in $\delta_j$ as fixed effects, i.e. they are not assumed to arise from a random effects distribution. This is both for simplicity and because in our teenager application the interest is on a few treatments and sub-groups, hence the fixed effects can be estimated accurately and provide a non-parametric alternative over assuming a particular random effects distribution (e.g. Normal). In situations where one wishes to specify random effects, these can in principle be incorporated into our framework by setting a suitable hierarchical prior on $(\beta_j,\delta_j)$, but we omit such discussion as the actual implementation details and computational algorithms would require certain adjustments that might obscure the exposition.

\subsection{Bayesian model selection and averaging} \label{sec:methods_bma}

BSCA first uses Bayesian model selection to assign a score (posterior probability) to each model. 
Recall that each model $s$ defines a different set of non-zero entries in the parameters $(\alpha, \gamma, \eta, \beta, \delta)$ in (\ref{eq:problem_sca}), hence defining a specific configuration of controls to be used, treatments to be included and (potentially) whether heterogeneous treatment effects across sub-populations are needed or not.
Then, one aggregates results across models using Bayesian model averaging.
A short introduction to the Bayesian regression framework is provided in the Supplementary Material.
See also \citet{madigan_bms_1994} for a description of BMS and \citet{hoeting_bma_1999} for a tutorial on BMA.

The posterior probability of each model $s$ given by BMS is obtained from Bayes’ rule as
$$ p(s \mid y) = \frac{p(y \mid s)p(s)}{p(y)}= \frac{p(y \mid s) p(s)}{\sum_{s'=1}^S p(y \mid s') p(s')},$$
where $p(s)$ is a user-specified prior model probability and $p(y \mid s)$ is the so-called integrated likelihood (or marginal likelihood, or evidence) for model $s$. 
It can be computed using standard methods, either via closed-form expressions (when available), 
Laplace approximations \citep{kass_1990}, 
approximate Laplace approximations \citep{rossell_2021approximate} when $n$ or $p$ are large, or Markov Chain Monte Carlo (MCMC) methods \citep{friel_2012}.

To avoid contentious prior choices and simplify the calculation of $p(y \mid s)$, in BSCA we use the approximation
$p(y \mid s) p(s) \approx e^{-\frac{1}{2} \text{EBIC}_s}$,
where $\text{EBIC}_s$ is the extended Bayesian information criterion for model $s$ \citep{chen_extended_2008}.
Briefly, the EBIC selects the same model as computing $p(s \mid y)$ exactly under a so-called unit information prior on the regression parameters and a Beta-Binomial prior on the models, under fairly general conditions as the sample size $n$ grows (\citealt{schwarz_estimating_1978}; \citealt{rossell_concentration_nodate}, Section 3). 
This gives approximate posterior probabilities
$$ p(s \mid y) \approx \frac{e^{-\frac{1}{2} \text{EBIC}_s}}{\sum_{s'=1}^{S} e^{-\frac{1}{2} \text{EBIC}_{s'}}}.$$

BMA then uses these posterior probabilities to obtain weighted estimates that
formally acknowledge the uncertainty in what is the right set of control variables, and the potential existence of treatment effect heterogeneity across subpopulations.
Specifically, BMA expresses the posterior distribution of the parameters as the weighted average
$$
p(\alpha, \beta, \gamma, \delta, \eta \mid y)= \sum_{s=1}^S p(\alpha, \beta, \gamma, \delta, \eta \mid s, y) p(s \mid y),
$$
where $p(\alpha, \beta, \gamma, \delta, \eta \mid s,y)$ is their posterior distribution under model $s$.
When the number of models $S$ is too large to enumerate all models we use Markov Chain Monte Carlo methods, specifically a Gibbs sampling algorithm implemented in the mombf R package used in our examples. 
Point estimates and 95\% intervals for average treatment effects $\beta_j$ in (\ref{eq:ate}) or subgroup-specific treatment effects in (\ref{eq:te}) are obtained as the mean and 95\% interval of their BMA posteriors $p(\beta_j \mid y)$ and $p(\beta_j + \delta_j^T g_i \mid y)$ respectively, also estimated by MCMC.

The BMA estimate and 95\% interval take into account the uncertainty arising from the many possible model specifications.
We emphasize that a main motivation for the original SCA was that standard errors conditional on a single selected model fail to account for the model selection uncertainty and selective reporting \citep{simonsohn_specification_2020}. This issue is resolved in BSCA, as one considers all controls and subpopulations.

Importantly, the EBIC incorporates a type of control for false discoveries under multiple hypothesis testing. Intuitively, it uses a more stringent threshold to include parameters into the selected model as the number of parameters $p$ grows (e.g. one has more treatments, controls or subgroups).
For example, the EBIC threshold to drop a parameter from the full model including all terms in (\ref{eq:problem_sca}) is roughly given by a P-value threshold that is of order $1/np^2$.
For comparison, this threshold is more stringent than the $1/p$ threshold used by a Bonferroni P-value adjustment.
Further, as $n$ grows under mild conditions the EBIC selects the true model $s^*$ that only selects the truly non-zero parameters, even when $p$ far exceeds $n$ \citep{chen_extended_2008}. 
This ensures that family-wise type I and II error probabilities both converge to 0, and in particular that one is protected from selective reporting.
In fact, a stronger property holds: the posterior probability of the true model $p(s^* \mid y)$ given by the EBIC converges to 1 under mild conditions for fixed $p$ \citep{schwarz_estimating_1978,dawid_2011}.
The result also holds, under slightly stronger conditions, in Gaussian regression where $p$ grows faster than $n$
even when data are not exactly generated by the assumed regression model – e.g. due to omitted covariates, non-linear effects, or other incorrect parametric assumptions \citep{rossell_concentration_nodate}.

The main practical implication of these theoretical results is that, when $n$ is large (e.g. our teenager application), the BMA point estimate and 95\% interval for treatment effects converge to the frequentist MLE and 95\% confidence interval obtained under the optimal model (i.e. the true model $s^*$, if there is no model misspecification).
Such a property is critical in applied settings where there is no theoretical framework (e.g. given by psychological or behavioral theory) guiding which controls or subpopulations one should consider.
See the conclusions for a discussion on how to interpret the framework when the assumed model is misspecified. 

\subsection{Single-outcome BSCA} \label{sec:methods_single}

\begin{figure}[t]
  \includegraphics[width=\linewidth]{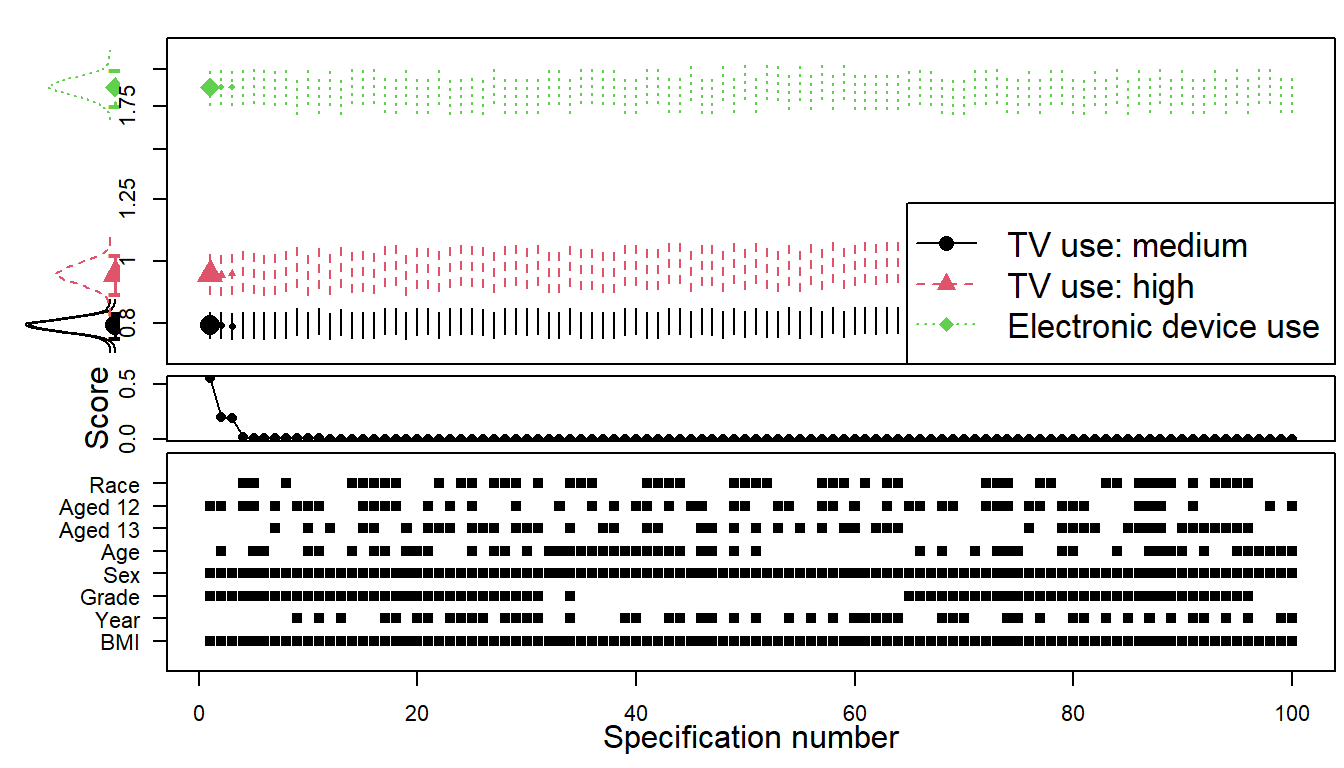}
  \caption{Single-outcome BSCA. Top-right: estimated effect of TV and electronic device use on the probability of thinking about suicide with 95\% interval for different treatments and models, marker size proportional to model score. Middle: EBIC-based model scores. Bottom: included controls. Top-left: Bayesian Model Averaging point estimate, 95\% interval and posterior distribution. Effects are odds-ratios compared to no usage for 1-4 hours (medium) and \textgreater4 hours (high) of TV usage and \textgreater=5 hours of Electronic device use. All 1024 possible confounder combinations were considered. Data source: YRBS.}
  \label{fig:bsca_single_yrbs}
\end{figure}

The main output of BSCA when focusing on a single outcome is a plot portraying treatment effect estimates under each considered model (or those with highest posterior probability, when there are too many models to display), as well as the BMA posterior distribution, the aggregated point estimate and 95\% posterior interval.

Figure \ref{fig:bsca_single_yrbs} shows a single-outcome BSCA using the data described in Section \ref{sec:application}.
The BSCA top right panel gives point estimates and 95\% intervals given by the posterior distribution of the treatment effects $p(\beta \mid s,y)$ for the 100 models with highest posterior probability $p(s \mid y)$. This panel is analogous to standard SCA, except that we focus attention on the models more supported by the data.
The BSCA middle panel gives the model scores, that is the posterior model probabilities $p(s \mid y)$.
The idea is that, by looking at these scores, one gets an assessment of what models are supported by the data.
The bottom panel mimicks that in SCA and indicates the covariates included by each model.
Finally, the BSCA top left panel displays the BMA posterior distribution $p(\beta_j \mid y)$, the posterior mean $E(\beta_j \mid y)$ as a point estimate, and its 95\% posterior interval. 

The BMA posterior distribution is also used to perform a hypothesis test for $\beta_j=0$, specifically when the posterior probability $P(\beta_j \neq 0 \mid y)$ exceeds a threshold one rejects $\beta_j=0$. For simplicity, we recommend the threshold $P(\beta_j \neq 0 \mid y)>0.95$.
Moreover, our EBIC-based formulation guarantees that, if truly an effect $\beta_j=0$ is not present, then the expectation of $P(\beta_j \neq 0 \mid y)$ and the type I error rate converge to 0, as $n$ grows (\citealt{rossell_concentration_nodate}, Proposition 1 and Corollaries 1-2).
The simulation study in Section \ref{sec:simulation} illustrates these properties

Note that the top panel in Figure \ref{fig:bsca_single_yrbs} shows the ATE in (\ref{eq:ate}), i.e. the effect of each treatment on a given outcome, averaged across subpopulations if these were included in the models.
If subgroup-specific treatment effects in (\ref{eq:te}) are desired, they can either be shown in the top panel as additional curves (in our R function this can be achieved with the $coefidx$ parameter) or in a separate figure.

\subsection{Multiple-outcome BSCA} \label{sec:methods_multiple}

When considering multiple outcomes, e.g. several assessments related to teenager well-being, it is desirable to summarize all the treatment-outcome pairs in a single plot, so that one can easily assess their heterogeneity.
This can be easily achieved by reporting the BMA summaries (top left panel in Figure \ref{fig:bsca_single_yrbs}) in a single plot.
As an illustration, Figure \ref{fig:bsca_multiple} shows an example with 5 treatments and 8 outcomes.
We call such a plot the multiple-outcome BSCA.

Besides assessing the significance of each treatment-outcome pair separately, 
there are situations where researchers may want to assess the global average treatment effect (GATE) across multiple treatments and/or outcomes.
This may be particularly suitable when the treatments correspond to the use of similar technologies (e.g. two social media platforms), or when the outcomes are expected to measure similar underlying characteristics (e.g. feeling depressed in two consecutive months).
We next discuss how to obtain such a GATE.

Suppose that one has $L \geq 1$ outcomes and $J \geq 1$ treatments. Let $\beta_{jl}$ be the regression coefficient associated to treatment $j \in \{1,\ldots,J\}$ and outcome $l= \{1,\ldots,L\}$.
One may then define the global effect of treatment $j$ across all outcome as
\begin{align}
 \mbox{GATE}_j= \frac{1}{L} \sum_{l=1}^L \beta_{jl}
\nonumber
\end{align}
and the global effect across all treatments and outcomes as
\begin{align}
\mbox{GATE}= \frac{1}{JL} \sum_{l=1}^L \sum_{j=1}^J \beta_{jl}.
\label{eq:gate}
\end{align}
BMA point estimates for the $\mbox{GATE}_j$ and $\mbox{GATE}$ above are given by the posterior means
\begin{align}
E(\mbox{GATE}_j \mid y) &= \frac{1}{L} \sum_{l=1}^L E( \beta_{jl} \mid y),
\nonumber \\
E(\mbox{GATE} \mid y) &= \frac{1}{JL} \sum_{l=1}^L \sum_{j=1}^J E( \beta_{jl} \mid y),
\nonumber
\end{align}
and one may similarly obtain a 95\% posterior interval from $p(\mbox{GATE}_j \mid y)$ and $p(\mbox{GATE} \mid y)$.

As a practical concern, such posterior intervals in principle require one to formulate a multivariate regression model for the $L$ outcomes that captures the dependence between outcomes, which results in posterior dependence for the $\beta_{jl}$'s. 
Although feasible, applying BMA to such multivariate models requires a more involved implementation and computation.

Fortunately, for continuous outcomes (e.g. Gaussian regression) a simpler strategy is possible.
It suffices to use a univariate regression (\ref{eq:problem_sca}) where the response is defined as the average across the $L$ outcomes, which we denote by $m_i= \sum_{l=1}^L y_{il}/L$, and then immediately inference for the $\mbox{GATE}_j$ and the GATE in (\ref{eq:gate}).
Proposition \ref{thm:ate} provides a precise statement, and follows from simple algebra.

\begin{prop}\label{thm:ate}
  Let $y_i$ be an $L$-dimensional response, $x_i$ a $J$-dimensional treatment, $z_i$ a $Q$-dimensional control covariate vector and $g_i$ a $K$-dimensional subpopulation-membership vector for individuals $i=1,\ldots,n$. 
  Consider the regression
  \begin{equation}\label{eq:model_linear}
    y_i = \alpha + \gamma z_{i} + \eta g_{i} + \sum_{j=1}^J \left( \beta_{j} x_{ij} + \delta_{j} x_{ij} g_{i} \right) + \epsilon_i,
  \end{equation}
  where $\epsilon_i \in \mathbb{R}^L$ is a zero-mean error vector,
  $\alpha \in \mathbb{R}^L$ the intercept, 
  $\gamma$ an $L \times Q$ matrix with the control regression coefficients,
  $\eta$ an $L \times K$ matrix with the subgroup main effects,
  $\beta = (\beta_1, \ldots, \beta_J)^T$ an $L \times J$ matrix containing the treatment effects, 
  and $\delta_j$ an $L \times K$ matrix with their modifiers for the sub-populations.
  
  Let $m_i= \sum_{l=1}^L y_{il}/L$ be the mean of all outcomes in $y_i$. Then
  $$
  m_i= \tilde{\alpha} + \tilde{\gamma}^T z_{i} + \tilde{\eta}^T g_{i} + \sum_{j=1}^J \left( \tilde{\beta}_{j} x_{ij} + \tilde{\delta}_{j}^T x_{ij} g_{i} \right) + \xi_i,
  $$
  where $\xi_i= \sum_{l=1}^L \epsilon_{il}$, $\tilde{\alpha}= \sum_{l=1}^L \alpha_l/L$,
  $\tilde{\gamma}$ is a $Q \times 1$ vector with $q^{th}$ entry equal to $\sum_{l=1}^L \gamma_{lq}/L$,
 $\tilde{\eta}$ a $K \times 1$ vector with $k^{th}$ entry equal to $\sum_{l=1}^L \eta_{lk}/L$,
  $\tilde{\delta}_j$ is a $K \times 1$ vector with $k^{th}$ entry equal to $\sum_{l=1}^L \delta_{jlk}/L$,
 and
  $$\tilde{\beta}_j= \frac{1}{L} \sum_{l=1}^L \beta_{jl}= \mbox{GATE}_j.$$

Further,
  $$
  \mbox{GATE}= \sum_{j=1}^J \frac{\tilde{\beta}_j}{J}.
  $$
\end{prop}

\begin{proof}
  See Appendix \ref{sec:proofs}.
\end{proof}

An important remark, which makes us caution against using the GATE for statistical inference, is that it assigns equal weight to all outcomes. This may be inappropriate in situations where some of the outcomes are strongly correlated. For instance, suppose that there are $J=10$ outcomes, 9 of which measure a very similar latent quantity (they are similar items within a questionnaire) whereas the tenth outcome measures an inherently different quantity. Then the GATE will be mostly determined by outcomes 1-9, whereas intuitively one might want to discount their weight. 
That is, generally speaking it is unclear that the GATE measures a particularly sensible parameter.
Given that defining an alternative GATE is a potentially contentious issue, and that as discussed treatment effects can be highly heterogeneous across outcomes, by default we recommend reporting inference for each treatment-outcome combination separately, and to only consider the GATE in circumstances where it is clear from the application that it is a sensible parameter.

\subsection{Alternative measures of association}
\label{ssec:alternative_measures}

Suppose that, rather than reporting regression coefficients $\beta$, a researcher wants to consider another treatment-outcome association measure $\rho$.
For example, if both $y$ and $x$ are continuous variables, one may want to report partial correlations.
Our framework can easily accommodate any measure $\rho$ that can be described as a function of model parameters $\theta$, i.e. $\rho= g(\theta)$. The BMA posterior distribution $p(\theta \mid y)$ implies a BMA posterior $p(\rho \mid y)$, i.e. one may immediately report BSCA inference on $\rho$.

As an illustration, suppose that $(y,x,z)$ follow a continuous distribution and, for simplicity, that there are no moderators.
The partial correlation between the outcome $y$ and treatment $x_j$, given the controls $z$ and all other treatments $x_{-j} = (x_1, \ldots, x_{j-1}, x_{j+1}, \ldots, x_p)$, is defined as
$$ \rho_j= cor(y, x_j | x_{-j}, z) = -\beta_j \sqrt{\frac{var(x_j | y, x_{-j}, z)}{var(y|x, z)}},$$
where $\beta_j$ is the linear regression coefficient in \eqref{eq:model_linear} (setting the link function $F$ to the identity).
Let $\mbox{var}(y \mid x, z)= \phi I$ where $\phi >0$ is the error variance, assumed equal for all $i=1,\ldots,n$. 
BMA gives a posterior distribution $p(\beta, \phi \mid y)$, and hence also on $\rho_j$, provided one has an estimate of $\mbox{var}(x_j \mid y, x_{-j},z)$. For example, in our illustrations in Section S2.6 we estimated the latter with the residual variance from a linear regression of $x_j$ on $(y,x_{-j},z)$.
An alternative is to postulate a Bayesian regression model for $x_j$ given $y, x_{-j},z$, which would provide a full BMA posterior on $(\beta,\phi,\mbox{var}(x_j \mid y,x_{-j},z))$, and hence on $\rho$. The latter strategy is however computationally more burdensome (one must do a model-fitting for each treatment), and is hence not considered here.

\section{Simulation study} \label{sec:simulation}

We illustrate the statistical properties of SCA and BSCA via several simulation studies. 
Section \ref{ssec:sim_singletreat} considers two settings with $J=1$ treatments and $L=1$ outcomes, shows issues with inflated type I error, bias and root mean squared error (RMSE) for SCA, and how BMA solves these issues.
Section \ref{ssec:sim_multitreat} considers two settings with $J=6$ treatments and $L=1$ outcomes, and Section \ref{ssec:sim_multioutcome} with $J=5$ treatments and $L=4$ outcomes. These serve to illustrate that, as predicted by the theory discussed in Section \ref{sec:methods}, despite conducting multiple tests our BMA implementation provides a good control of the type I error.
In these three sections we simulate both the case where there truly is no treatment effect and there truly is an effect.
The BSCA test is based on the posterior inclusion probability $P(\beta_j \neq 0 \mid y) > 0.95$, as described in Section \ref{sec:methods_bma}.
The SCA hypothesis test is based on the bootstrap-based P-value proposed by \citet{simonsohn_specification_2020}, using a 0.05 significance level.
For comparison we also considered a permutation-based P-value, which gave very similar results and is not reported here.

In all settings we used $n=1,000$ observations, a single control $z_i \in \mathbb{R}$ that has an effect on the outcome and is also correlated with the $J$ treatments, and we report average results across 100 independent simulations.
Specifically, the data-generating process is
\begin{align*}
  y_i & = \beta^T x_i + z_i + \epsilon_i, \\
  z_i & \sim N\left( \sum_{j=1}^J x_{ij} /J, 1 \right) \\
  x_i & \sim N(0, I)
\end{align*}
independently across $i=1,\ldots,n$, where $\beta \in \mathbb{R}^J$ are the true treatment effects and $\epsilon_i \sim N(0,1)$. In Section \ref{ssec:sim_multioutcome} where $y_i \in \mathbb{R}^4$ is a multivariate outcome, we used correlated errors $\epsilon \sim N(0, \Sigma)$. Specifically, $\Sigma$ has unit diagonal, pairwise correlations equal to 0.9 among outcomes 1-3, and 0.1 correlation with outcome 4. This correlation structure is meant to represent a situation where the first three outcomes can be thought of as measuring one common latent characteristic, that is different from that measured by the fourth outcome.

We chose a relatively large $n=1,000$ because the teenager studies also had large $n$ (actually larger). In such settings the data-generating model typically receives a high posterior probability, and in fact in the different analysis of the teenager data there was always a model with high posterior probability.
The simulations show that in such settings one attains high power and low type I error, in fact in all the examples below BSCA did not incur any type I nor any type II error, whereas we show that for SCA these errors had a probability near 1.
The reason is that the large $n$ does not improve the properties of the SCA median, which is an asymptotically biased estimator of the data-generating $\beta$.
We provide our R code so that readers can assess performance in other settings.

\begin{table}[bt]
\centering
\begin{tabular}{rllrr}
  \hline
 & Scenario & Estimator & Bias & RMSE \\ 
  \hline
1 &  & BMA & -0.001 & 0.010 \\ 
  2 & \multirow{-2}*{$\beta = (0)$} & SCA & 0.498 & 0.499 \\ 
  3 &  & BMA & -0.003 & 0.044 \\ 
  4 & \multirow{-2}*{$\beta = (1)$} & SCA & 0.497 & 0.498 \\ 
  5 &  & BMA & 0.000 & 0.002 \\ 
  6 & \multirow{-2}*{$\beta = (0, 0, 0, 0, 0, 0)$} & SCA & 0.075 & 0.077 \\ 
  7 &  & BMA & -0.003 & 0.012 \\ 
  8 & \multirow{-2}*{$\beta = (0, 0, 0.25, 0.75, 1, 1)$} & SCA & 0.055 & 0.072 \\ 
  9 & GATE = 0 & BMA & 0.000 & 0.002 \\ 
  10 & GATE $\neq$ 0 & BMA & -0.001 & 0.009 \\ 
   \hline
\end{tabular}
\caption{Bias and root mean squared error in simulation scenarios 1-5. Scenarios 1-2 have a single treatment, Scenarios 3-4 have 6 treatments, and Scenarios 5-6 have 4 treatments and 5 outcomes. $\beta$ indicates the data-generating truth.} 
\label{tab:sim_summary}
\end{table}

\subsection{Single treatment} \label{ssec:sim_singletreat}

\paragraph{Truly zero treatment effect}

We start with a single treatment. In our first scenario the treatment has a truly zero effect, and there is a single control with a non-zero effect. This scenario assesses the type I error, that is the frequentist probability that BMA would wrongly claim the treatment effect to exist. Rows 1 and 2 of Table \ref{tab:sim_summary} show the estimated bias and RMSE for the BMA and SCA median estimators, respectively. The BMA estimate has a bias close to zero, in contrast SCA tends to over-estimate the true parameter value $\beta_1=0$. The reason is that the treatment $x$ is correlated with the control $z$, which truly has an effect, hence the model including $x$ but not $z$ over-estimates $\beta_1$ (this follows from simple algebra and standard least-squares theory).

Regarding the type I error, recall that the BMA test rejects the null hypothesis $\beta_1=0$ when the posterior probability $P(\beta_1 \neq 0 \mid y) > 0.95$. In all simulations $P(\beta \neq 0 \mid y)$ took a small value, hence the null hypothesis was never rejected and the estimated type I error is 0.
In contrast, the estimated type I error for SCA was 1. The reason is that SCA tests the median effect across all covariate specifications. In our case there are two such specifications, depending on whether one includes the control or (wrongly) excludes it from the analysis. Although the median across these two specifications is non-zero, the true treatment effect is zero.

\paragraph{Truly non-zero treatment effect}

We repeat the exercise with a non-zero treatment effect $\beta_1=1$. Rows 3-4 of Table \ref{tab:sim_summary} report the estimated bias and RMSE. Using BMA, the null hypothesis was rejected in all simulations, hence the estimated power is 1.

\subsection{Multiple treatments} \label{ssec:sim_multitreat}

\paragraph{Truly no treatment effect}

\begin{table}[bt]
\centering
\begin{tabular}{rrrrrrr}
  \hline
 & x1 & x2 & x3 & x4 & x5 & x6 \\ 
  \hline
Bias & -0.000 & 0.000 & 0.001 & 0.000 & -0.000 & 0.001 \\ 
  RMSE & 0.003 & 0.006 & 0.009 & 0.003 & 0.003 & 0.009 \\ 
  Type I error & 0.000 & 0.000 & 0.000 & 0.000 & 0.000 & 0.000 \\ 
   \hline
\end{tabular}
\caption{Individual coefficient results for BMA in simulation \#3 with coefficients $\beta = (0, 0, 0, 0, 0, 0)$} 
\label{tab:sim_3_indiv}
\end{table}

Moving on to multiple treatments, we first consider a setting where 6 treatments truly have a zero effect, so the ATE=0.
Table \ref{tab:sim_summary} shows the bias and RMSE for the ATE estimate, again BMA provides a significant reduction for both relative to SCA.
Table \ref{tab:sim_3_indiv} shows the bias and RMSE for individual treatments. BMA did not declare as significant any individual treatment in any of the simulated datasets, that is both the individual and family-wise type I error probabilities are estimated to be near-zero.  The individual treatment effects are not usually considered in SCA, and are hence not reported in Table \ref{tab:sim_3_indiv}.

\paragraph{Truly non-zero treatment effect}

Next, we consider a setting where 2 treatments truly have no effect ($\beta_1 = \beta_2 = 0$) and 4 treatments have heterogeneous effects ($\beta_3=0.25$, $\beta_4=0.75$, $\beta_4=1$, $\beta_5=1$).
We consider effects of different magnitudes, to illustrate the power to detect smaller versus larger effects.
Again, we include one control covariate that is correlated with the treatments and truly has an effect.
Note that the true ATE and median treatment effects are both 0.5. This is to facilitate comparison between BMA and SCA, since in our implementation BMA targets the ATE and SCA the median.

\begin{table}[bt]
\centering
\begin{tabular}{rrrrrrr}
  \hline
 & x1 & x2 & x3 & x4 & x5 & x6 \\ 
  \hline
Bias & -0.001 & -0.002 & -0.002 & -0.006 & -0.004 & -0.001 \\ 
  RMSE & 0.010 & 0.016 & 0.030 & 0.034 & 0.033 & 0.035 \\ 
  Proportion rejected & 0.000 & 0.010 & 1.000 & 1.000 & 1.000 & 1.000 \\ 
   \hline
\end{tabular}
\caption{Individual coefficient results for BSCA in simulation \#4 with coefficients $\beta = (0, 0, 0.25, 0.75, 1, 1)$} 
\label{tab:sim_4_indiv}
\end{table}

Here, BMA correctly detected that the ATE$\neq 0$ in all simulations. 
The bias and RMSE associated to the ATE are reported in rows 7-8 of Table \ref{tab:sim_summary}.  
In Table \ref{tab:sim_4_indiv} we report the bias, RMSE and type I error probabilities for individual treatments. BMA correctly detected that $\beta_1=\beta_2=0$, and that $\beta_3 \neq 0$, $\beta_4 \neq 0$, $\beta_5 \neq 0$ and $\beta_6 \neq 0$ in all simulations.

\subsection{Multiple outcomes} \label{ssec:sim_multioutcome}

As discussed, although we generally recommend running BSCA for each outcome individually and reporting the whole heterogeneity across outcomes, in some situations one may also be interested in a global ATE across all $L$ outcomes and $J$ treatments.

Our first simulation considers a setting where none of the treatments truly has an effect, i.e. $\beta_{jl} = 0$ for all $j=1,\ldots,4$ and $l=1,\ldots,5$, so that the $\mbox{ATE}=0$.
The estimated type I error rate is zero for BMA.
The bias and RMSE to estimate the $\mbox{GATE}$ are reported in row 9 of Table \ref{tab:sim_summary}.

Next we consider a case where 3 treatments truly have an effect, whereas treatments 4-5 do not. The effect of treatments 1-3 on outcomes 1-3 is different to that on outcome 4, mimicking a situation where outcomes 1-3 measure a common latent characteristic. 
Specifically, we used
$$
\beta=
\begin{pmatrix}
1.00 & 1.00 & 1.00 & 0.25 \\ 
1.00 & 1.00 & 1.00 & 0.25 \\ 
1.00 & 1.00 & 1.00 & 0.25 \\ 
0.00 & 0.00 & 0.00 & 0.00 \\ 
0.00 & 0.00 & 0.00 & 0.00 \\ 
\end{pmatrix}
$$
In all simulations the posterior inclusion probability was $>0.95$, i.e. the estimated power for the Bayesian test was 1.

\section{Teenager well-being and technology} \label{sec:application}

For our application of BSCA to the debate on teenager well-being and technology, we use two datasets: the Youth Risk Behavior Survey (YRBS) and Millennium Cohort Study (MCS) data. The YRBS is a biennial survey of adolescent health risk using a representative sample of US secondary school students. We use data from 2007-2015 ($n=75083$). The MCS is a socio-economic and health survey tracking a representative sample of children born in the UK around year 2000 throughout their life. We use data from the 2015 survey ($n=11884$).
Section \ref{ssec:data} describes the data.

Section \ref{ssec:teen_multitreat} applies BSCA to the YRBS data. We find that TV and electronic device usage have opposite associations (negative and positive, respectively) with thinking about suicide. Averaging over these treatments, as done by the SCA median, results in a near-zero association. 
Section \ref{ssec:teen_multioutcome} applies BSCA to the MCS data, which shows that social media usage has opposite associations with parent-assessed and self-assessed well-being (positive and negative, respectively). Averaging over these outcomes again results in a near-zero association, despite there being sizable heterogeneous effects. 
Finally, Section \ref{ssec:teen_subgroups} considers a gender-based subgroup analyses, which reveals important differences between girls and boys in the association between social media usage and well-being.

\subsection{Data and model} \label{ssec:data}

The data were prepared as described by \citet{orben_association_2019}. 
We adjusted the definition of some of the variables by expressing them on a common scale to facilitate interpretation of the model parameters, using validated psychometric scales as opposed to (unvalidated) individual outcome variables, fixing minor errors related to variable codes and including unemployed parents in the analysis.
Whenever the outcome is a binary variable, we use a logistic regression model. 
See Section S4 for further details.

In the YRBS data we focus the discussion on the binary outcome indicating whether the subject reported having thought about suicide. Additional outcomes considered in our analysis are: felt lonely, planned suicide, attempted to commit suicide, and saw a doctor due to an attempted suicide (all self-reported).
The treatment variables are time spent watching TV as well as time spent using electronic devices – such as video consoles, computers, smartphones and tablets – for things other than school work.
In all YRBS regressions, we control for the following control covariates recorded in the study: race, sex, grade, year of the survey and body-mass index (BMI).
We discretized TV usage into low, medium and high, since our pre-analysis revealed a non-monotonic association. All treatment variables had monotonic associations with the outcome variables. The code for these data pre-processing results is available in our repository at \url{https://doi.org/10.17605/OSF.IO/M8D4N}.

The MCS data is more extensive. Here we look at the outcomes self-assessed depressive symptoms (Mood and Feelings Questionnaire) and self-esteem (Rosenberg scale), as well as parent-assessed total difficulties, emotional problems, conduct problems, hyperactivity/inattention, peer problems and pro-sociality (all Strengths and Difficulties Questionnaire).
To facilitate the comparison of the treatment effects relative to the YRBS data, 
we defined a binary version of these outcomes using validated cutoffs for abnormal behavior (Section S4.2). For completeness, Section S3.5 repeats the analysis using a linear model for the original (non-binary) outcomes, which reproduces all our main findings.
The treatment variables are time spent watching TV, playing video games, using the internet and using social media, as well as owning a PC.
In all MCS regressions, we control for sex, age, BMI, (self-reported) educational motivation, mother's ethnicity, (self-reported) closeness to parents, presence of natural father, time spent with primary caretaker (PC), PC's word activity score, PC's employment status, own longstanding illness, PC's psychological distress, number of siblings and household income.

For both datasets, the treatment variables (except TV usage, see above) were treated as continuous variables, normalized such that 0 means “no usage” and 1 means maximum reportable usage. The most usage respondent’s can report is 5 or more hours in the YRBS dataset and 7 or more hours in the MCS dataset. The coefficient $\beta_j$ from logistic regression models therefore gives the log-odds ratio of maximum usage relative to 0 hours. We report the corresponding odds ratio (i.e.\ $e^{\beta_j}$).

\subsection{Analyzing multiple treatments} \label{ssec:teen_multitreat}

The \emph{single-outcome BSCA} in Figure \ref{fig:bsca_single_yrbs} shows estimated associations between different technologies and the probability that a teenager thinks about suicide. TV and electronic device usage had \emph{opposite} associations with suicidal thoughts. Using BMA, higher TV usage is \emph{not} associated with higher odds relative to low TV use (moderate use actually has lower odds), whereas $\geq$5 hours of electronic device use is associated with 1.88 (95\% CI: 1.76--2.00) times \emph{higher} odds of thinking about suicide (compared to no usage).
Their sign and magnitude are the same for the outcomes feeling lonely, planning suicide, attempting suicide and seeing a doctor about suicide (Figure S1).

This example illustrates a problematic issue with the original SCA.
Averaging over technologies, which in this study are TV and electronic device use, gives an ATE that may appear practically irrelevant. Such an ATE likely misled \cite{orben_social_2019} to their conclusion that technology has no relevant association with well-being, whereas we argue that 1.88 higher odds of thinking about suicide are definitely practically relevant.
In fact, in more recent work the same authors studied the effect of individual treatments \citep{orben_social_2019}. 

This example also illustrates that model scores (posterior probabilities) were highly informative about which controls are important. In Figure \ref{fig:bsca_single_yrbs} all top models include gender and school grade, neither of which were considered by \citet{orben_association_2019}.
These controls are quantitatively important. Being a boy is associated with 0.45 (95\% CI: 0.43-0.47) lower odds of thinking about suicide.
Moreover, students in early high school years have worse predicted mental health. Among the school years in the sample (grades 9 through 12), moving up one grade is associated with 0.92 (95\% CI: 0.86-0.94) lower odds of thinking about suicide.
BMA allows to perform such inference on the coefficients of the controls, which can reveal interesting information beyond the treatment effects of primary interest.

\subsection{Analyzing multiple outcomes} \label{ssec:teen_multioutcome}

\begin{figure}[t]
  \includegraphics[width=\linewidth]{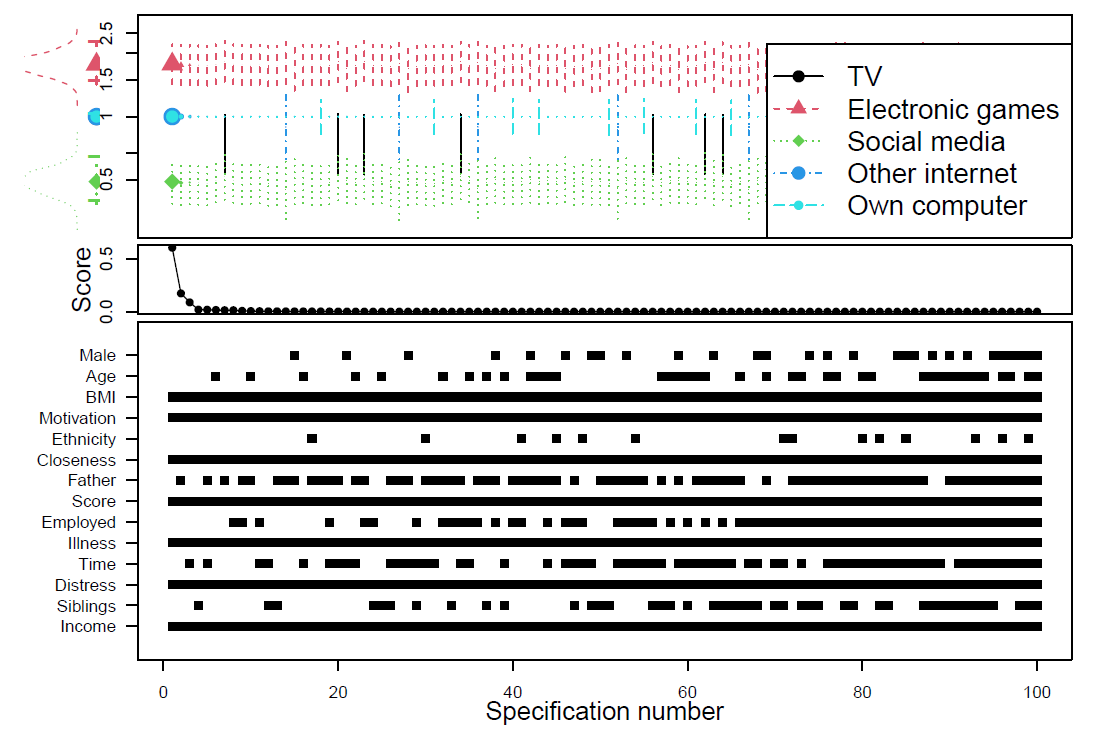}
  \caption{Single-outcome BSCA. Top-right: estimated effect of different technology uses on the probability of the outcome high total problems (parent-assessed, Strengths and Difficulties Questionnaire score $\geq$ 14) with 95\% interval for different treatments and models, marker size proportional to model score. Middle: EBIC-based model scores. Bottom: included controls. Top-left: Bayesian Model Averaging point estimate, 95\% interval and posterior distribution. Effects are odds-ratios for >7 hours compared to no usage. Data source: MCS.}
  \label{fig:bsca_single_mcs}
\end{figure}

When considering multiple outcomes, it is important to explore if treatment effects are heterogeneous across outcomes. 
In particular, if said effects are of a different sign, then focusing inference on GATEs may be severely misleading.
Such an issue occurs in the teenager data.
Figure \ref{fig:bsca_single_mcs} reveals that using social media for $\geq$7 hours is associated with \emph{lower} odds of parent-assessed total difficulties, compared to no usage (odds-ratio 0.49, 95\% CI: 0.38-0.62). This finding is in contrast with the negative effect of social media on mental well-being found in experiments \citep{allcott_welfare_2020}. A deeper analysis reveals that social media is actually associated with more \emph{adolescent-assessed} difficulties.

\begin{figure}[t]
  \includegraphics[width=\linewidth]{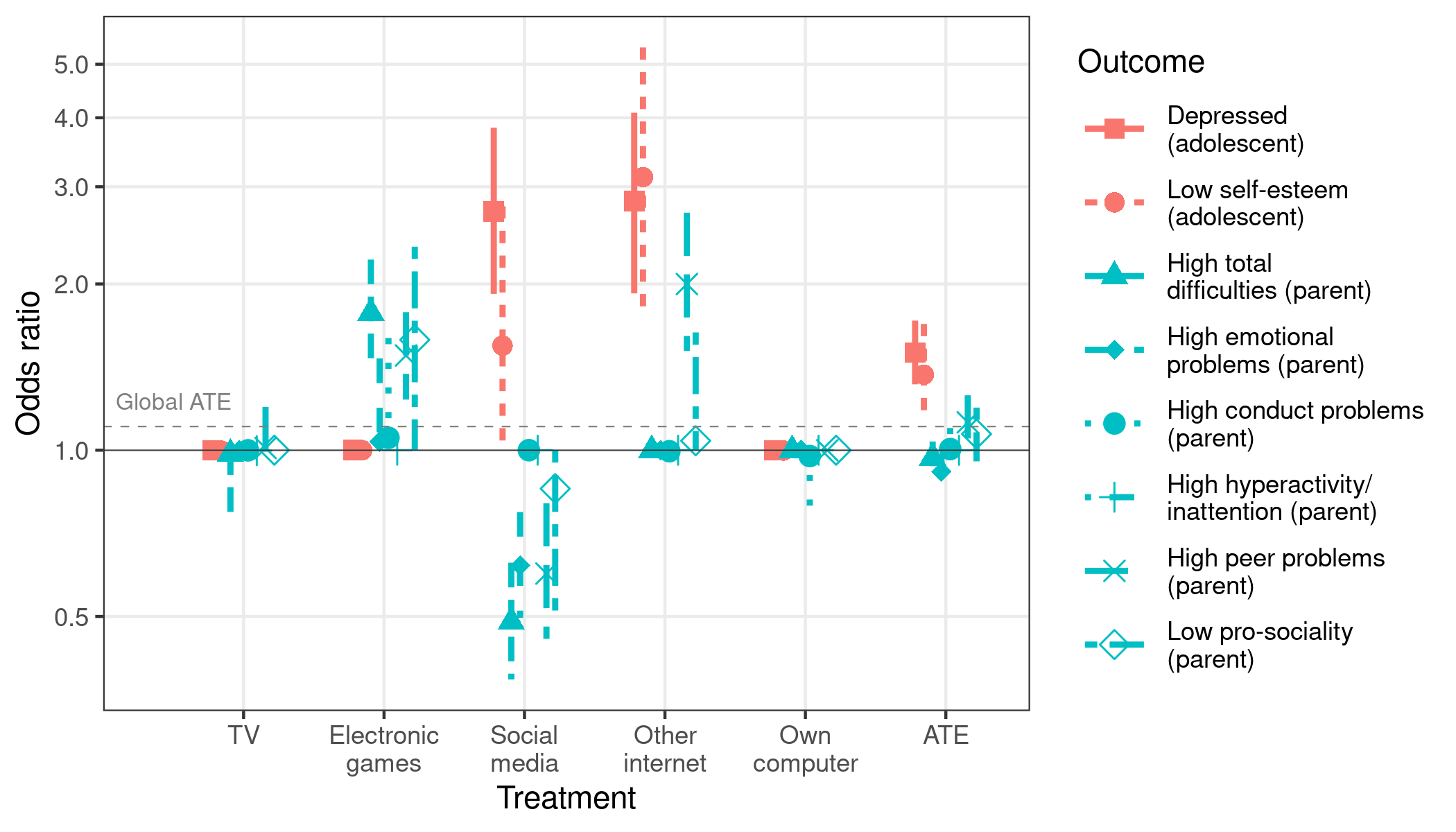}
  \caption{Multiple-outcome BSCA. BMA estimates and 95\% intervals for odds ratios between no use and \textgreater=7 hours of usage of different technologies and the average treatment effect (ATE) on several measures of mental well-being. The dashed horizontal line indicates the global ATE across all outcomes. Outcome variables: self-reported depression (Mood and Feelings Questionnaire short version score $\geq$ 12) and low self-esteem (Rosenberg short version score $\leq$ 7); parent-reported high total problems, high emotional problems, high conduct problems, high hyperactivity/inattention, high peer problems and low pro-sociality (Strengths and Difficulties Questionnaire score $\geq$ 14, $\geq$ 4, $\geq$3, $\geq$ 6, $\geq$ 3 and $\leq$ 5). Controls: age, sex, BMI, educational motivation, mother's ethnicity, closeness to parents, presence of natural father, time spent with primary caretaker (PC), PC's word activity score, PC's employment status, longstanding illness, PC's psychological distress, number of siblings and household income. Data source: MCS.}
  \label{fig:bsca_multiple}
\end{figure}

To identify these issues, the \emph{multiple-outcome BSCA} in Figure \ref{fig:bsca_multiple} summarizes all single-outcome analyses into one display. This plot jointly reports all analyses, and hence avoids selective reporting. Figure \ref{fig:bsca_multiple} reveals that social media use is associated with \emph{lower} odds of total and emotional problems according to parents, but 2.72 (95\% CI: 1.92--3.85) and 1.52 (95\% CI: 1.00-2.75) times \emph{higher} odds of depressive symptoms and low self-esteem, respectively, according to adolescents.
That is, while parents assessed less difficulties for teenagers who use social media heavily, the assessment of the teenagers themselves had an opposite sign.

Figure \ref{fig:bsca_multiple} also includes the ATE for each outcome, illustrating that it is near-zero for outcomes with opposing individual treatment effects, such as parent-assessed total difficulties (which has a positive association with electronic games, and a negative one with social media). Finally, Figure \ref{fig:bsca_multiple} shows a global ATE across outcomes, but we caution against its use for formal inference. Besides potentially averaging over inherently different outcomes, the global ATE may over-weight blocks of highly correlated outcomes (see Section \ref{sec:methods_multiple}).

\subsection{Analysing subgroups} \label{ssec:teen_subgroups}

\begin{figure}[t]
  \includegraphics[width=\linewidth]{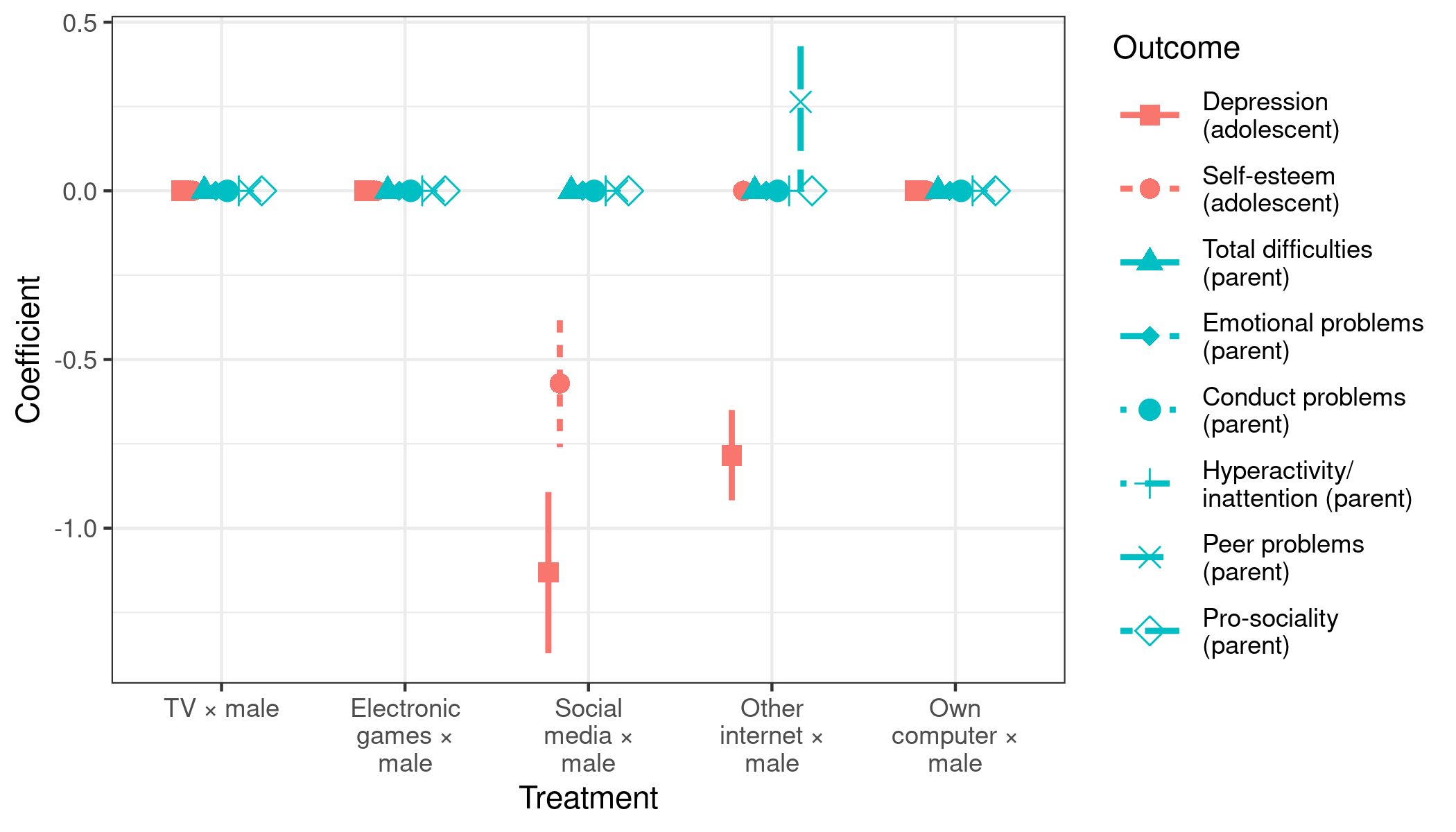}
  \caption{Multiple-outcome BSCA for gender subgroups. BMA estimates and 95\% intervals for treatment-gender interaction terms from a linear regression model measuring the ATE in females - the ATE in males. Outcomes are the raw scores of the variables in Figure \ref{fig:bsca_multiple}, normalized to 1-10. Treatments are scaled to lie in [0,1], where 0 means no use and 1 means >=7 hours of usage. Control covariates are the same as in Figure \ref{fig:bsca_multiple}. Data source: MCS}
  \label{fig:bsca_subgroup}
\end{figure}

There is some evidence in the literature on social media usage and mental health that there may be relevant interactions with gender: girls with high exposure to social media are more likely to be depressed than boys with the same amount of exposure \citep{kelly_social_2018,twenge_why_2020}.
To study this issue, we perform an analysis where we consider two sub-populations defined by gender.
We consider the 5 treatments in the MCS data and the 8 outcomes measuring well-being considered in  Section \ref{ssec:teen_multioutcome},
treated as continuous variables given by the survey scores.
For completeness, we repeated the analyses after defining a binary version of all outcomes (as described in Section \ref{ssec:teen_multioutcome}). In this case the interaction vs. gender was not statistically significant, likely due to a loss of statistical power for binary-outcome relative to continuous-outcome regression. See Section S3.6 for further discussion.

The results are shown in Figure \ref{fig:bsca_subgroup}. High social media usage is indeed associated with an average -1.1 points lower depression and a -0.8 points lower self-esteem score for males compared to females (scale: 1-10). Other internet usage is also associated with a -0.6 points lower self-assessed depression score among males compared to females and, interestingly, 0.25 points higher association with parent-assessed peer problems. Again, there is a discrepancy between the assessments made by teenagers and their parents.
For all other outcomes and treatments, including the use of TV and electronic devices and owning a computer, BSCA did not find evidence for an interaction with gender. That is, a subgroup analysis is not justified for these treatment-outcome pairs.

\section{Discussion} \label{sec:discussion}

We introduce the Bayesian Specification Curve Analysis (BSCA) as a tool which, besides visualizing the sensitivity of results to the chosen statistical model, provides a strategy to identify the most promising models and aggregate over control configurations, while allowing one to detect heterogeneity across treatments, outcomes and subpopulations. BSCA can either be used to assess treatment effects on a single outcome, or to summarize the association with multiple outcomes.
We adopted a formulation based on the Extended Bayesian Information Criterion (EBIC) that is simple and helps prevent false positives. 
It also helps bypass a common critique of Bayesian statistics that the specification of priors could be seen as arbitrary: our posterior model probabilities are a monotone transformation of the EBIC, a popular criterion in the non-Bayesian community.

Using BSCA, we find that technology use has an association with teenager well-being of high statistical and practical significance. Some technologies are associated with suicidal thoughts, self-reported depression and low self-esteem. We also find that social media has opposite associations with parent and adolescent assessments. 
This heterogeneity is masked when taking the median over all estimated effects. By displaying such heterogeneous treatment effects, BSCA helps avoid potentially misleading conclusions.
Some may disagree with portraying such heterogeneity, arguing that estimates for different outcomes compatible with the researcher’s theory should be always be aggregated. 
We showed that, if so desired, BSCA allows to report and test such aggregated estimates (global ATE).
However, we argue that one should always assess heterogeneity. If there is clear heterogeneity – that, in the extreme case seen in our application, cancels out in the aggregate – practitioners may need to revise their theory or at least report the underlying heterogeneity. By providing tests for individual treatment effects, BSCA accommodates this requirement.

Importantly, data-based aggregation methods like SCA and BSCA cannot replace theory to select outcomes and covariates. Theory is essential to generate well-specified models and predictions \citep{muthukrishna_problem_2019}. Instead, these methods are useful when a given theory is ambiguous about which variables should be used. If the researcher’s theory requires that some controls or sub-populations should definitely be included, that can easily be done in BMA by setting their prior inclusion probabilities to 1. That is, Bayesian methods offer a strategy to combine theory with data-based evidence. Like the the selection of the general model, this should always be reported, justified, and pre-registered if possible.

It is important to remark that our findings show the existence of a conditional association between technology and teenager well-being, and in particular cannot be used to establish a causal connection between them. 
We do note, however, that our findings are fairly compatible with existing literature and available causal evidence.
Previous studies have found different associations (positive, negative and null) between technology use and adolescent mental well-being and a small negative association in the aggregate \citep{stiglic_effects_2019,orben_teenagers_2020,twenge_why_2020}. This is compatible with our finding that different kinds of technology use have opposite associations that average to a small negative effect, and highlights our message that one should not blindly average over treatments and outcomes.
Moreover, observational evidence \citep{valkenburg_social_2022} and causal evidence from randomized controlled trials \citep{allcott_welfare_2020,allcott_digital_2022} and natural experiments \citep{braghieri_social_2022} establish a negative effect of social media on self-reported mental well-being, as in Figure \ref{fig:bsca_multiple}. Some of these studies also support that the effect is stronger for girls \citep{allcott_welfare_2020}, as in Figure \ref{fig:bsca_subgroup}.
Our finding that the association of self-assessed and parent-assessed well-being with technology has not been documented to the best of our knowledge. It could be driven by the somewhat different measures (e.g., psychosomatic complaints and behavioral difficulties) or by the disagreement between self and parent reports \citep{roberts_concordance_2005,chen_parent-child_2017,poulain_parent-child_2020}. The latter has been shown to lead to differential associations between well-being and atopic diseases \citep{keller_atopic_2021}. The finding deserves a more detailed analysis that falls outside the scope of this paper, which is primarily methodological.

In line with our recommendations for conducting a BSCA, we assess heterogeneity and robustness with respect to many analytic decisions flagged in the literature, including datasets, treatments, control covariates (Figures \ref{fig:bsca_single_yrbs} and \ref{fig:bsca_single_mcs}), outcomes (Figures \ref{fig:bsca_multiple} and S1), moderators (Figure \ref{fig:bsca_subgroup}), functional form and variable dichotomization (Section S3) as well as non-linear dose-response relationship (pre-analysis file). Due to data limitations, we cannot vary some other contentious choices, including the cross-sectional design, self-reported screen time as the treatment measure, and leaving out several important moderators (for example, we cannot check for for difference between active and passive users, private and public usage, or types of self-presentation).

The improvement of statistical properties of BSCA over the median effect reported by SCA is due to weighing models using their probability given the data. 
Taking an unweighted average is ``a degenerate form of Bayesian model averaging in which we never update our priors'' \citep{slez_difference_2019}. 
In Bayesian terms, the SCA median is a prior predictive model where all models are weighted equally. As a result, the estimator does not generally converge to the true effect as $n$ grows. By contrast, BSCA uses posterior weights informed by the data, i.e. takes into account the probability that the data where generated by any one model.
As discussed, our EBIC-based formulation guarantees that, as the sample size $n$ grows, under mild conditions the total number of false discoveries converges to zero. 

We remark, however, that these theoretical results need to interpreted carefully in practice, because models are simplifications of reality and hence are often misspecified. One may misspecify the structure for the mean, the error distribution, or the covariance structure. In treatment effect estimation misspecifying the mean structure is the most serious of these issues, in the sense that even as $n$ grows one fails to select the right covariates and recover the true treatment effect. For example, one may fail to record truly relevant control covariates, specify wrongly the functional form of their effect, or that of the treatments. The latter two issues can be addressed with exploratory data analysis and model-checking diagnostics, as we did in our analyses. However, if relevant controls were not recorded, one unavoidably runs into omitted variable biases. For this reason, one should keep in mind that the reported treatment effects do not establish causal relationships, but are a measure of conditional association. For example, we found that technology use is conditionally associated with several teenager well-being outcomes after accounting for age, gender and socio-economic covariates, but including further covariates (e.g. past depression history) might change these estimates.

Alternative strategies are of course possible. One could consider other Bayesian formulations that guarantee false positive control, e.g. the so-called complexity priors on the model space of \cite{castillo_2015} or the non-local priors on the regression coefficients \citep{johnson_2010,johnson_2012,rossell_2017}. 
It is also possible to refine our formulation for situations where one considers many outcomes. Briefly, the EBIC guarantees that the family-wise error rate converges to 0 for large $n$ for each individual outcome. Hence, for any fixed number of outcomes $L$, the total false positives across outcomes also converges to 0. If one considered a very large $L$, one could easily adapt the EBIC so that it is based on the total number of parameters across all outcomes, and therefore provide a false positive control.
While these refinements are potentially interesting, we found them to be unnecessary in the considered applications.
Finally, one could also use non-Bayesian false positive control methods. 
BMA could be viewed simply as a mechanism to obtain a point estimate for either an ATE or a single treatment effect $\hat{\beta}_{jl}= E(\beta_{jl} \mid y)$. Given such point estimates, one can use permutation tests to obtain P-values and standard P-value adjustment or False Discovery Rate control methods to control for multiple testing. See \citet{benjamini_1995,efron_2007} for discussion on FWER and FDR control. 
These methods are typically less stringent than the EBIC in that they target a non-zero probability of including some false positives, whereas for the EBIC said probability converges to 0 as $n$ grows.

\section*{Acknowledgements}

We thank Profs.\ Orben \& Przybylski for providing their code and helpful feedback. 
C.S. was supported by ''la Caixa'' Foundation (ID 100010434) grant LCF/BQ/DR19/11740006. D.R. was partially supported by Ramon y Cajal grant RYC-2015-18544, Spanish Government grant PGC2018-101643-B-I00, and Europa Excelencia grant EUR2020-112096.

\section*{Supplementary Materials}

The Supplementary Material 
includes an introduction to Bayesian regression, instructions on how to reproduce our results, robustness checks, and data treatment details.
The code to reproduce our empirical analysis and produce the Bayesian Specification Curve Analysis (BSCA) plots can be found in our Open Science Framework repository at \url{https://doi.org/10.17605/OSF.IO/M8D4N}.

\appendix

\section{Proof of Proposition \ref{thm:ate}}\label{sec:proofs}

First, note that the global ATE can be written as
$$
\mbox{GATE}= \sum_{l=1}^L \sum_{j=1}^J \frac{\beta_{jl}}{JL}= \frac{{\mathbb 1}_J^T}{J}  \beta \frac{{\mathbb 1}_L}{L} 
$$
where ${\mathbb 1}_J=(1,\ldots,1)^T$ is the $J$-dimensional one-vector, ${\mathbb 1}_L$ the $L$-dimensional one-vector and recall that $\beta$ is an $J \times L$ matrix. Since $m_i= \sum_{l=1}^L y_{il}/L= {\mathbb 1}_L^T y_i/L$, it follows that
$$
m_i = \frac{{\mathbb 1}_L^T}{L} \alpha + \frac{{\mathbb 1}_L^T}{L} \gamma z_{i} + \frac{{\mathbb 1}_L^T}{L} \eta g_{i} + \sum_{j=1}^J \left( \frac{{\mathbb 1}_L^T}{L} \beta_{j} x_{ij} + \frac{{\mathbb 1}_L^T}{L} \delta_{j} x_{ij} g_{i} \right) + \frac{{\mathbb 1}_L^T}{L} \epsilon_i 
= \tilde{\alpha} + \tilde{\gamma}^T z_{i} + \tilde{\eta}^T g_{i} + \sum_{j=1}^J \left( \tilde{\beta}_{j} x_{ij} + \tilde{\delta}_{j}^T x_{ij} g_{i} \right) + \xi_i
$$
where $\tilde{\alpha} = \frac{{\mathbb 1}_L^T}{L} \alpha= \sum_{l=1}^L \alpha_l/L$, 
$\tilde{\gamma}^T = \frac{{\mathbb 1}_L^T}{L} \gamma$, 
$\tilde{\eta}^T = \frac{{\mathbb 1}_L^T}{L} \eta$,
$\tilde{\beta}_j = \frac{{\mathbb 1}_L^T}{L} \beta_j$,
$\tilde{\delta}_j^T = \frac{{\mathbb 1}_L^T}{L} \delta_j$,
and $\xi_i = \frac{{\mathbb 1}_L^T}{L} \epsilon_i= \sum_{l=1}^L \epsilon_{il}/L$.
Hence,
$$
\mbox{GATE}= \frac{{\mathbb 1}_J^T}{J}  \beta \frac{{\mathbb 1}_L}{L}= \frac{{\mathbb 1}_J^T}{J} \tilde{\beta}= \sum_{j=1}^J \frac{\tilde{\beta}_j}{J},
$$
as we wished to prove.

Finally, note that $\tilde{\gamma} = \frac{1}{L} \gamma^T {\mathbb 1}_L$ is a $Q \times 1$ vector with $q^{th}$ entry given by $\sum_{l=1}^L \gamma_{lq}/L$,
$\tilde{\eta} = \frac{1}{L} \eta^T {\mathbb 1}_L$ a $K \times 1$ vector with $k^{th}$ entry given by $\sum_{l=1}^L \eta_{lk}/L$
and $\tilde{\delta}_j = \frac{1}{L} \delta_j^T {\mathbb 1}_L$ a $K \times 1$ vector with $k^{th}$ entry equal to $\sum_{l=1}^L \delta_{jlk}/L$.

\bibliographystyle{plainnat}
\bibliography{specurveanalysis}

\end{document}